\documentclass[11pt]{article} 
\usepackage[margin=1in]{geometry}
\usepackage{pgfplots}

\usepackage{tikz}
\usepackage{amsthm}

\usepackage{amsmath}
\usepackage{amssymb}
\usepackage{latexsym}

\usepackage{hyperref}
\hypersetup{
    colorlinks=true,
    linkcolor=violet,
    filecolor=magenta,      
    urlcolor=cyan,
    citecolor=blue
    ,
    pdffitwindow=true,
}\usepackage[capitalize, nameinlink]{cleveref}
\usepackage{enumerate}
\usepackage{accents}
\usepackage{pstricks}
\usepackage{centernot}
\usepackage{microtype}
\usepackage{verbatim}
\usepackage{geometry}  
\usepackage{float}
\usepackage{subcaption}

\usepackage{pgf, tikz}
\usetikzlibrary{arrows, automata}

\newtheorem*{theorem*}{Theorem}

\newtheorem{lemma}{Lemma}[section]
\newtheorem{theorem}[lemma]{Theorem}

\newtheorem{claim}{Claim}

\newtheorem{problem}{Problem}

\crefname{claim}{Claim}{Claims}
\Crefname{claim}{Claim}{Claims}

\renewcommand{\P}[1]{{\mathbb{P}}\left[#1\right]}
\renewcommand{\Pr}{\mathbb{P}}
\usepackage{pgfplots}
\usetikzlibrary{intersections}
\usepgfplotslibrary{fillbetween}

\newcommand{\DED}{{\sc DED($k$)}}
\newcommand{\DVD}{{\sc DVD($k$)}}

\newcommand{\R}{\mathbb{R}}

\newcommand{\E}{\mathbb{E}}

\usepackage{cleveref}

\usepackage{authblk}

\begin{document}
\title{A Randomized Rounding Approach for DAG Edge Deletion}
 \date{}
% \author{}

\author{Sina Kalantarzadeh} 
\affil{\small University of Waterloo\\ \textsf{S4kalant@uwaterloo.ca}}

\author{Nathan Klein} 
\affil{\small Boston University\\ \textsf{nklei1@bu.edu}}

\author{Victor Reis} 
\affil{\small Microsoft Research\\ \textsf{victorol@microsoft.com}}
 
 \maketitle

 \begin{abstract}
    In the DAG Edge Deletion problem, we are given an edge-weighted directed acyclic graph and a parameter $k$, and the goal is to delete the minimum weight set of edges so that the resulting graph has no paths of length $k$. This problem, which has applications to scheduling, was introduced in 2015 by Kenkre, Pandit, Purohit, and Saket. They gave a $k$-approximation and showed that it is UGC-Hard to approximate better than $\lfloor 0.5k \rfloor$ for any constant $k \ge 4$ using a work of Svensson from 2012. The approximation ratio was improved to $\frac{2}{3}(k+1)$ by Klein and Wexler in 2016. 

    In this work, we introduce a randomized rounding framework based on distributions over vertex labels in $[0,1]$. The most natural distribution is to sample labels independently from the uniform distribution over $[0,1]$. We show this leads to a $(2-\sqrt{2})(k+1) \approx 0.585(k+1)$-approximation. By using a modified (but still independent) label distribution, we obtain a $0.549(k+1)$-approximation for the problem, as well as show that no independent distribution over labels can improve our analysis to below $0.542(k+1)$. Finally, we show a $0.5(k+1)$-approximation for bipartite graphs and for instances with structured LP solutions. Whether this ratio can be obtained in general is   open.
\end{abstract}

\newpage
\section{Introduction}

Given an edge-weighted directed acyclic graph $G=(V,A)$ and a parameter $k$, the DAG Edge Deletion problem of $k$, or \DED, is to delete a minimum weight set of edges so that the resulting graph has no paths with $k$ edges. \DED~was first introduced by Kenkre, Pandit, Purohit, and Saket \cite{Kenkre2015}, who gave a simple $k$-approximation and showed it was UGC-Hard to approximate better than $\lfloor \frac{k}{2} \rfloor$ for any constant $k$ using a result on the vertex deletion version of the problem by Svensson \cite{Svensson2012}. Kenkre et al. also studied the maximization version of \DED, known as Max $k$-Ordering, which can be studied on general directed graphs. In Max $k$-Ordering, we want to \textit{keep} as many edges as possible so that there are no paths of length $k$ remaining. They gave a 2-approximation for Max $k$-Ordering. Max $k$-Ordering and \DED~are NP-Hard for $k \ge 2$ \cite{Lampis2008,Klein2015}.

For \DED, it is important that the input is a DAG. Otherwise, it becomes impossible to approximate when $k$ depends on $n$ (which it may), as it can encode the Directed Hamiltonian Path problem, and requires prohibitive running time when $k$ is large.\footnote{Although, somewhat surprisingly, one can find a path of length $k$ in time $2^k \cdot \text{poly}(n,k)$ \cite{Williams2009}.} However, on a DAG, finding paths of arbitrary length is polynomial time solvable using a dynamic program. For Max $k$-Ordering, it is not necessary to assume the input is a DAG since one can still obtain an approximate solution by comparing to an optimal solution that may keep all the edges. In this work we will focus on the regime in which our running time cannot depend exponentially on $k$ and $k$ may depend on $n$. 

\DED~is identical to the following problem: find a partition $V_1,\dots, V_k$ of $G$ such that the weight of edges that do not go from a smaller indexed component to a larger one is minimized. This was observed by \cite{Kenkre2015}, and follows from the fact that a DAG with no paths of length $k$ has a topological ordering with at most $k$ layers. This gives a natural scheduling (or parallel computing) application for \DED: if the vertices represent tasks and the edges represent constraints that task $i$ must occur before task $j$, \DED~is equivalent to finding the minimum weight set of constraints to violate such that all the tasks can be completed within $k$ serial steps. \DED~is also motivated by a problem in VSLI design, similar to the DAG vertex deletion version (\DVD) introduced by Paik \cite{Paik1994} in the 90s. For \DVD, we are given a circuit as input and the goal is to upgrade the nodes of the circuit so that no electrical signal has to travel through $k$ consecutive un-upgraded  nodes. \DED~can be motivated identically, except here we upgrade the wires instead of the nodes. \DVD~also has been studied for its applications to secure multi-party computation, see for example \cite{CCMN25}.

After the work of Kenkre et al. introducing \DED, giving a linear program for it, and showing a $k$-approximation, Klein and Wexler \cite{Klein2015} showed a $\frac{2}{3}(k+1)$-approximation using a combinatorial rounding approach, as well as improved results for small values of $k$.  
They also showed an integrality gap of at least $\frac{k+1}{2}$ for the natural LP using a construction of Alon, Bollob{\'a}s, Gy{\'a}rf{\'a}s, Lehel, and Scott \cite{ABGLS07}. 
In \cref{sec:improved_alg}, we show the following improved approximation ratio and integrality gap:

\begin{theorem}\label{thm:main}
    There is a randomized approximation algorithm for \DED~with approximation ratio $0.549(k+1)$. In particular, given a solution $x$ to the natural LP \eqref{eq:DED_LP}, there is a randomized algorithm producing a solution of expected cost at most $0.549(k+1)\cdot c(x)$, where $c(x)$ is the cost of the LP. Therefore, the integrality gap of LP \eqref{eq:DED_LP} is at most $0.549(k+1)$. 
\end{theorem}
Note that it is easy to derandomize our algorithm up to an arbitrary loss in precision using the method of condition expectation. We will not discuss this as it is standard.

Our algorithm rounds using a distribution over vertex labels in $[0,1]$. The most natural distribution to pick is uniform over the interval $[0,1]$, which leads to an approximation ratio of $(2-\sqrt{2})(k+1) \approx 0.585(k+1)$. The above theorem, giving a factor of $0.549(k+1)$, is obtained using a distribution which has more mass near 0 and 1. A natural question is what the optimal distribution over labels is. While we are unable to answer this question precisely, in \cref{sec:lowerbound}, we show that there is no independent distribution obtaining better than $0.542(k+1)$ using our framework. The lower bound has some similarity to recent work in additive combinatorics on Sidon sets, where one needs a lower bound on the supremum of the probability density function of autoconvolutions \cite{CS17,MV10}.  

We also show a $0.5(k+1)$-approximation for \textit{bipartite} DAGs using a correlated labeling distribution. As the UGC-Hardness result of Kenkre et al. \cite{Kenkre2015} holds for bipartite graphs, this rounding is optimal up to constants under the UGC. Finally, we show a $0.5(k+1)$-approximation for structured LP solutions using an independent label distribution. In particular, if $x$ is an optimal LP solution for an instance and there exists some $c > 0$ so that $x_e \in \{0,c\}$ for all $e \in A$, then we show a randomized $0.5(k+1)$-approximation. (In reality, we show something slightly stronger than this, see \cref{sec:extras} for more details.) This rounding is optimal, as the integrality gap example in \cite{Klein2015} has this property. 

\subsection{Hardness Results and Other Related Work}

In 2012, Svensson~\cite{Svensson2012} gave lower bounds for DAG Vertex Deletion. In particular, he showed that for any constant $k$, it is hard to approximate \DVD~with a ratio better than $k$ assuming the UGC. A $k$-approximation for this problem is easy to obtain, so the approximability of \DVD~is completely understood (assuming the UGC) for all constants $k$. When $k$ may depend on $n$, however, the problem is still not fully understood, even under the UGC, as the lower bound no longer holds. The vertex deletion version has also been studied on undirected graphs for small $k$ \cite{Bresar2011, Tu2011}. As for NP-Hardness, there is an approximation preserving reduction from Vertex Cover to DVD(2) and to DED(4), so \DVD~is NP-Hard to approximate better than $\sqrt{2}$ for $k \ge 2$, as is \DED~for $k \ge 4$. NP-Hardness of $\sqrt{2}$ for Vertex Cover was shown by Khot, Minzer, and Safra \cite{KMS17} in 2017, improving upon the classical result of Dinur \cite{Dinur2005}. 

Notice that the NP-Hardness result works for \textit{any} $k \ge 2$ for \DVD~and any $k \ge 4$ for \DED, where $k$ may depend on $n$, as given hardness for parameter $c$, one can add paths of length $k-c$ with infinite weight vertices or edges to every vertex. This reduction works for the UGC-Hardness results as well, meaning that for any constant $k$ and any $k' \ge k$ (where $k'$ may depend on $n$), DED($k'$) is still UGC-Hard to approximate better than $\lfloor \frac{k}{2}\rfloor$ (or $k$ for DVD($k'$)). In other words, both problems do not become easier as $k$ grows. 

As noted in \cite{Kenkre2015}, \DED ~is related to the well-studied Maximum Acyclic Subgraph Problem (MAS). Here we are given a directed graph (which is not a DAG) and asked to find an acyclic subgraph with as many edges as possible. Note that the maximization version of \DED, Max $k$-Ordering, is a generalization of MAS, recovering this problem when $k$ is set to $n$. Guruswami, Manokaran, and Raghavendra showed that MAS is UGC-Hard to approximate better than $2-\epsilon$ for any $\epsilon > 0$ \cite{GMR08}, and Kenkre et al. \cite{Kenkre2015} note that this implies UGC-Hardness of $2-o_k(1)$ for Max $k$-Ordering.
Also related is the Restricted Maximum Acyclic Subgraph (RMAS) problem, where each vertex has a finite set of possible integer labels, and the goal is to pick a valid labeling that maximizes the number of edges going from a smaller to a larger label in the assignment.  This problem was introduced by Khandekar, Kimbrel, Makarychev, and Sviridenko \cite{KKMS09} and its approximability was studied by Grandoni, Kociumaka, and Włodarczyk \cite{Grandoni15}, who gave a $2 \sqrt{2}$-approximation. RMAS clearly generalizes MAS. To the best our knowledge, a minimization version of RMAS on DAGs has not been studied. 

A more general setting is the $k$-hypergraph vertex cover problem for which a simple rounding gives a $k$-approximation. For the setting where the underlying hypergraph is also $k$-partite, a theorem of Lov\'asz (see~\cite{Aharoni1996}) gives a $0.5k$-approximation.

%Finally, DVD and DED have been studied in restricted settings, and can be solved in polynomial time on rooted trees and series parallel graphs \cite{Paik1994}.

\subsection{Formulation as a CSP}

\DED~can easily be stated as a constraint satisfaction problem. In particular, one can create a variable $x_v$ for each vertex $v \in V$ and allow $x_v \in \{1,2,\dots,k\}$. Then, we have a set of constraints given by the edges: for each $e=(u,v)$ we create a constraint that $x_u < x_v$. Max $k$-Ordering is then the problem of maximizing the number of satisfied constraints, and \DED ~is the problem of minimizing the number of unsatisfied constraints. So, both problems fall under the purview of Raghavendra's UGC-optimal algorithm for CSPs \cite{Ragha08} \textit{whenever $k$ is an absolute constant}. However, the approximation ratio of this algorithm for constant $k$ is not known, and the runtime of the rounding step depends exponentially on $k$. (In addition, for \DED, the objective function is negative, and there is an additive error in Raghavendra's algorithm which could produce issues for this approach.) So, while this is an important view on the problem, it does not give satisfactory guarantees as $k\to\infty$ or when $k$ depends on $n$. 

\subsection{Organization of the Paper}

In \cref{sec:improved_alg}, we show our improved approximation algorithm. Our approach is to assign labels $\ell(v) \in [0,1]$ to each vertex $v$ independently and cut edges $e=(u,v)$ for which $\ell(v)-\ell(u)$ is small compared to $x_e$. By choosing an appropriate definition of small, the feasibility of the LP solution then implies cutting these edges results in a feasible solution to \DED. We first show that drawing labels $\ell(v) \sim \text{Uniform}(0,1)$ results in a $(2-\sqrt{2})(k+1) \approx 0.585(k+1)$-approximation, and then show a modified distribution that results in an algorithm with ratio slightly below $0.549(k+1)$. 

In \cref{sec:lowerbound}, we show that no independent labeling distribution can improve the approximation ratio of our algorithm to below  $0.542(k+1)$, at least within our analysis framework, which bounds the probability each edge is cut by $\alpha x_e$ to obtain an approximation ratio of $\alpha$.

Finally, in \cref{sec:extras}, we give our two additional results showing $0.5(k+1)$-approximation algorithms for bipartite DAGs and instances with structured LP solutions. 

\section{Improved Approximation Algorithm}\label{sec:improved_alg}

The following is the LP from Kenkre et al.~\cite{Kenkre2015} which they used to achieve a $k$-approximation for edge-weighted DAGs, where $c_e$ gives the cost of the edge $e$ and $x_e$ denotes the relaxation of the decision variable on whether edge $e$ is cut.
\begin{equation}\label{eq:DED_LP}
\begin{aligned}
\min\;& \sum_{e\in A} c_e\,x_e, \\[6pt]
\text{s.t.}\;& \sum_{e\in P} x_e \;\ge\; 1, \quad &&\forall\text{ paths }P\text{ of length }k,\\[3pt]
& x_e \;\ge\; 0, \quad &&\forall\,e\in A.
\end{aligned}
\end{equation}

Although there are an exponential number of constraints, a separation oracle can be implemented in polynomial time, as the input graph is a DAG. Therefore by the ellipsoid method, this LP can be solved in polynomial time. Note we will use $c(x)$ to denote $\sum_{e \in A} c_ex_e$ and for a set of edges $S \subseteq A$, $c(S) = \sum_{e \in S} c_e$.

%\begin{conjecture*} Applying labels in $[k]$ uniformly at random to all vertices and performing the above procedure to each edge based on how many zeros lie on any path of length $k$ containing that edge is a $\frac{k+1}{2}$ approximation. 
%\end{conjecture*}
%We have been unable to find any graphs for which this algorithm fails. 

\subsection{Algorithmic Framework}

Our algorithm is as follows, and depends on a \textit{label distribution} $\mu: [0,1]^V \to \R_{\ge 0}$ over label sets $\ell \in [0,1]^V$. We will use $\ell(v)$ to denote the label of vertex $v$. 

First we solve \eqref{eq:DED_LP} to obtain an optimal solution $x$. Then we draw a label set $\ell$ according to $\mu$, and delete all edges $e=(u,v)$ with 
$$\ell(v) - \ell(u) \le (k+1)x_e - 1$$
We will focus on upper and lower bounds for distributions in which the label of every vertex is drawn independently from the same  fixed distribution over $[0,1]$. However, in \cref{sec:extras} we will use a distribution that is not independent. Before analyzing the performance of two label distributions, we show now that it results in a feasible solution. 

\begin{lemma} For any solution $x$ to \eqref{eq:DED_LP} and any set of labels $\ell \in [0,1]^V$, deleting all the edges with $\ell(v) - \ell(u) \le (k+1)x_e - 1$ results in a feasible solution.\end{lemma}
\begin{proof}
 Let $P$ be a path of length $k$ beginning at vertex $u$ and ending at $v$. By way of contradiction, suppose no edges in the path are deleted. Then:
\begin{align*}
\ell(v) & > \ell(u) +  \sum_{e \in P} [(k+1)x_e - 1]\\
& \ge \sum_{e \in P} [(k+1)x_e - 1] & \text{Since $\ell(u) \ge 0$}\\
& = \sum_{e \in P} (k+1)x_e - k  & \text{Since $P$ is of length $k$}\\
& \ge (k+1) - k & \text{Since $x$ is a feasible solution} \\
& = 1
\end{align*}
This gives a contradiction because $\ell(v) \in [0,1]$. Note the first inequality is strict because we keep an edge $e=(a,b)$ only if $\ell(b) - \ell(a) > (k+1)x_e - 1$. 
\end{proof}

To bound the approximation ratio of our algorithm, we will use the following.
\begin{lemma}\label{lem:approxfactor}
    Let $\mu$ be a labeling distribution and  $\alpha \ge \frac{1}{2}$. If for any edge $e=(u,v)$ and any $t \in [-1,1]$ we have
    $$\frac{\Pr[\ell(v)-\ell(u)\leq t]}{t+1} \le \alpha$$
    Then, using $\mu$ in the above algorithm produces a solution of expected cost at most $\alpha(k+1) \cdot c(x)$ and results in a randomized $\alpha(k+1)$-approximation.
\end{lemma}
\begin{proof}
Let $e=(u,v)$ be an edge. We will show that 
\begin{equation}\label{eq:lem22}
    \P{\ell(u) - \ell(v) \le (k+1)x_e - 1} \le (k+1)\alpha \cdot x_e
\end{equation}
So long as we have this, the claim follows, because if $S$ is the set of edges cut by the algorithm, then:
    \begin{align*}
        \E[c(S)] &= \sum_{e=(u,v) \in A} \P{\ell(u) - \ell(v) \le (k+1)x_e - 1}c_e \\
        &\le 
        \alpha(k+1) \sum_{e \in A} c_ex_e = \alpha(k+1)\cdot c(x)
    \end{align*}
    To obtain \eqref{eq:lem22}, set $t=(k+1)x_e - 1$. Clearly, $t \ge -1$. If $t \ge 1$, then $x_e \ge \frac{2}{k+1}$, but then \eqref{eq:lem22} is trivially satisfied. So, $t=(k+1)x_e - 1 \in [-1,1]$ and by the assumption,
    $$\frac{\Pr[\ell(v)-\ell(u)\leq (k+1)x_e - 1]}{(k+1)x_e} \le \alpha.$$
    Rearranging gives \eqref{eq:lem22}.
\end{proof}

\subsection{Uniform Distribution}\label{subsec:uniform}

Here we consider the most natural label distribution: for every vertex, we simply assign it a label uniformly at random from the interval $[0,1]$, independent of all other vertices. First we show the following lemma:
\begin{lemma}\label{lem:unif}
    Consider the labeling distribution $\mu$ which independently assigns every vertex a label uniformly at random from $[0,1]$. Then, for all $u,v\in V$, $t \in [-1,1]$,
    $$\frac{\Pr[\ell(v)-\ell(u)\leq t]}{t+1} \le 2-\sqrt{2}$$
\end{lemma}
\begin{proof}
    Let $\ell(u),\ell(v)$ be draws from $\text{Uniform}(0,1)$. Then, the CDF $F(x)$ of $\ell(v)-\ell(u)$ is as follows:
    $$F(t) = \begin{cases} \frac{1}{2}t^2+t+\frac{1}{2} & -1 \le t < 0 \\ -\frac{1}{2}t^2 + t + \frac{1}{2} & 0 \le t \le 1 \end{cases}$$
If $-1 \le t < 0$, then
$$\frac{F(t)}{t+1} = \frac{(t+1)(\frac{1}{2}t+\frac{1}{2})}{t+1} = \frac{1}{2}(t+1) \le \frac{1}{2}$$
If $0 \le t \le 1$, then
$$\frac{F(t)}{t+1} = \frac{-\frac{1}{2}t^2 + t+\frac{1}{2}}{t+1}$$
Computing the derivative of this expression with respect to $t$, we obtain $\frac{-\frac{1}{2}t^2-t+\frac{1}{2}}{(t+1)^2}$, and setting this to 0, we know that the maximizer must be either $-1-\sqrt{2}$, $-1+\sqrt{2}$, or an endpoint $t=0$ or $t=1$. Of the two roots, only $-1+\sqrt{2}$ is in the range $t \in [0,1]$. When evaluated this is $2-\sqrt{2}$. Evaluating at $t=0$ or $t=1$ we obtain $\frac{1}{2}$, so the maximizer is $2-\sqrt{2}$. 
\end{proof}
Using \cref{lem:approxfactor}, we obtain the following as a corollary:
\begin{theorem}
Assigning $\ell(v)$ to each vertex $v$ independently and uniformly from $[0,1]$ results in a solution with expected cost at most $(2-\sqrt{2})(k+1) \cdot c(x) \approx 0.585(k+1) \cdot c(x)$. By using an optimal solution $x$, we obtain a $(2-\sqrt{2})(k+1)\approx 0.585(k+1)$-approximation.
\end{theorem}

\subsection{Improving Upon Uniform}\label{subsec:improved_distribution}

While using the uniform distribution is a natural candidate, it is (perhaps somewhat surprisingly) not optimal. If we wanted to obtain an $0.5(k+1)$ approximation using our framework, it is not difficult to see that the CDF of $\ell(v)-\ell(u)$ must be equal to $\text{Uniform}(-1,1)$ for all $u,v$ with an edge between them with $x_e > 0$. If this were obtainable, we would have:
$$\P{\ell(v) - \ell(u) \le (k+1)x_e - 1} = \frac{k+1}{2}x_e$$
for all values of $x_e \in [0,\frac{2}{k+1}]$. (Edges of value greater than $\frac{2}{k+1}$ are cut with probability 1 for any label distribution.)

So, in some sense, \textbf{our goal is to design an independent distribution over labels so that the CDF of $\ell(v)-\ell(u)$ is ``closer" to $\text{Uniform}(-1,1)$.} Using numerical verification and exhaustive search, we arrived at the probability distribution $\mathcal{D}$ over $[0,1]$ with density $p(x)=\frac{2}{3}+\frac{23}{3}(1-2x)^{22}$. If $X,Y$ are sampled independently from $\mathcal{D}$, we obtain a CDF for which $X-Y$ is closer to $\text{Uniform}(-1,1)$ in the relevant respect. We show that using this distribution yields a better approximation ratio than uniform labels. Due to the complexity of the resulting expressions, we use numerical verification.

% The next result we will not prove here due to the tedious calculations involved. However, we have verified it using Mathematica (and it is the best of many   functions we tried). Note the shape: it seems the best functions push most of the weight (but not too much of it) close to 0 and 1.

\begin{figure}[H]
\centering
\includegraphics[scale=0.4]{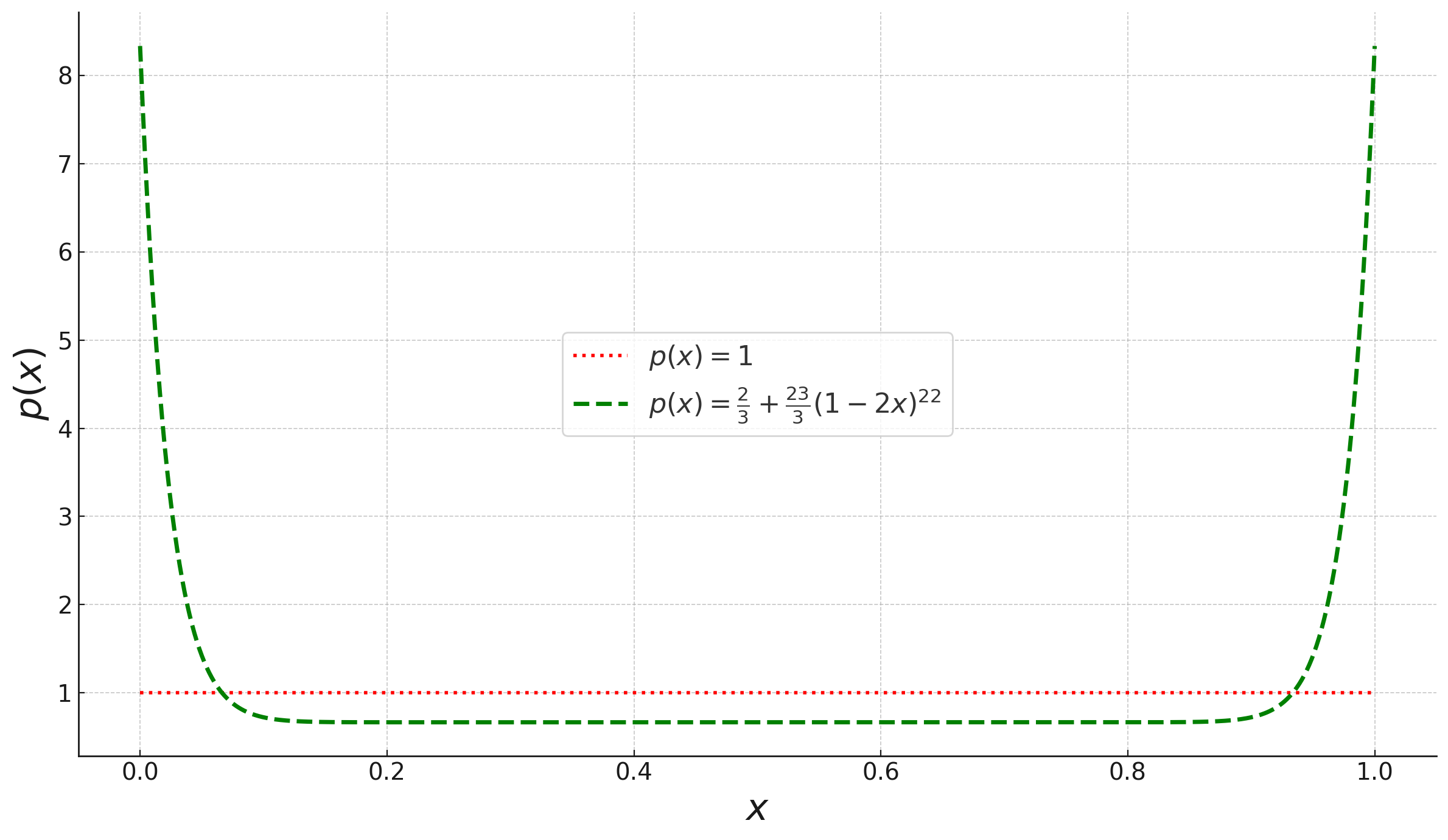}
\caption{Densities of $\mathcal{D}$ and $\text{Uniform}(-1,1)$ }
\end{figure}

\begin{theorem}
Assigning labels $\ell$ to all vertices from the interval $[0,1]$ independently according to the density $p(x) = \frac{2}{3}+\frac{23}{3}(1-2x)^{22}$ results in an a solution with expected cost below $0.549(k+1) \cdot c(x)$. This results in a randomized $< 0.549(k+1)$-approximation for \DED. 
\end{theorem}

As in the uniform case, to obtain the theorem it is sufficient to prove the following claim and apply \cref{lem:approxfactor}.
    \begin{claim}\label{Claim:MaxCrazyDist}
        Suppose $X,Y$ are independent random variables sampled from $\mathcal{D}$, then $\frac{\Pr[X-Y\leq t]}{t+1}< 0.549$ for all $t\in [-1,1]$.
    \end{claim}
    \begin{proof}
        This is verified by Mathematica. See \cref{sec:appendix} for the code used and an execution. The maximum of $\frac{\Pr[X-Y\leq t]}{t+1}$ occurs at approximately $t=0.27$ with value slightly below $0.5482$. \qedhere

    \end{proof}
    \begin{figure}[H]
  \centering
  \begin{subfigure}[b]{0.7\textwidth}
    \centering
    \includegraphics[width=\textwidth]{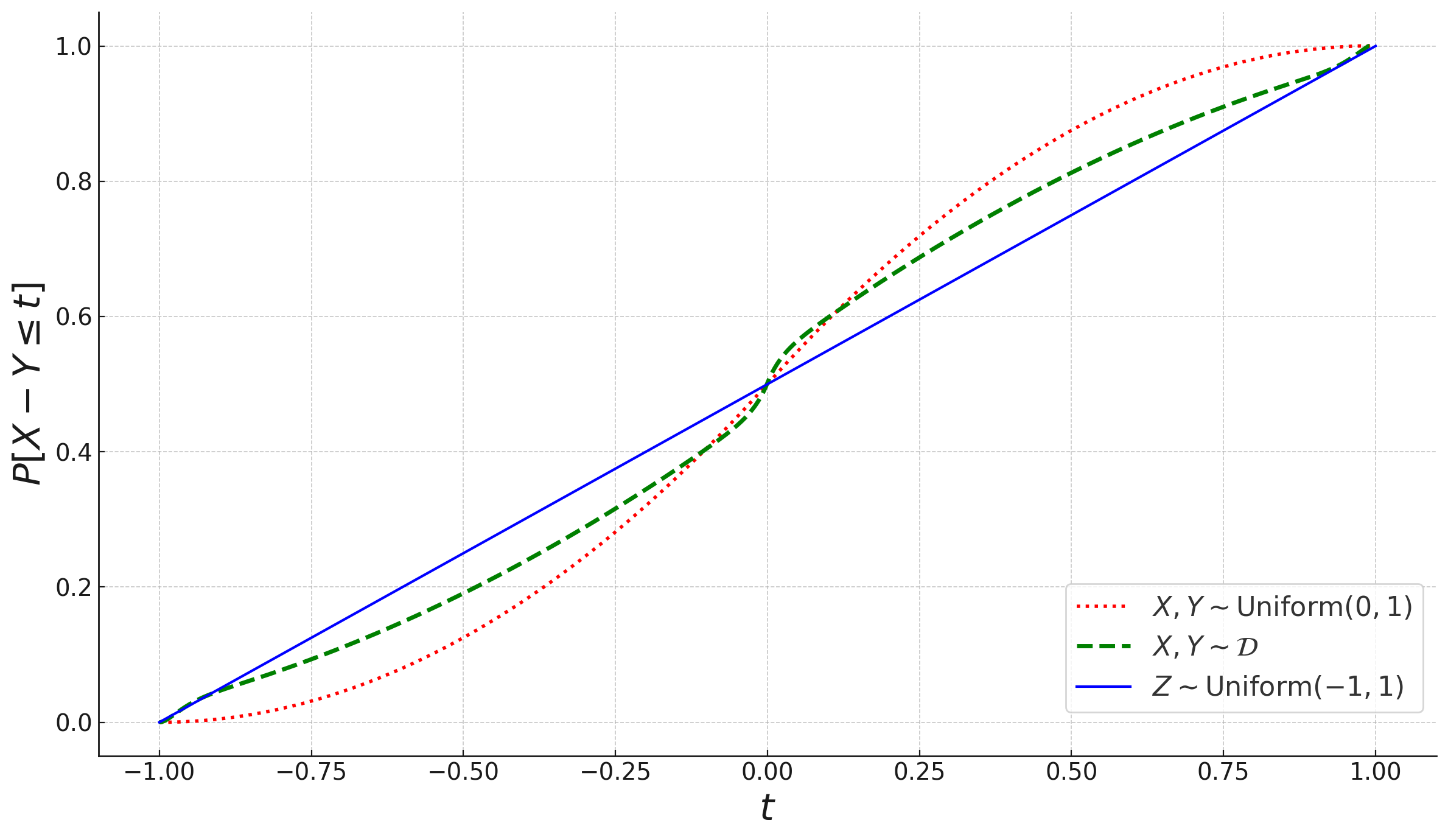}
    \caption{$\P{X-Y\leq t}$}
    \label{fig:2a}
  \end{subfigure}
  \hfill
  \begin{subfigure}[b]{0.7\textwidth}
    \centering
    \includegraphics[width=\textwidth]{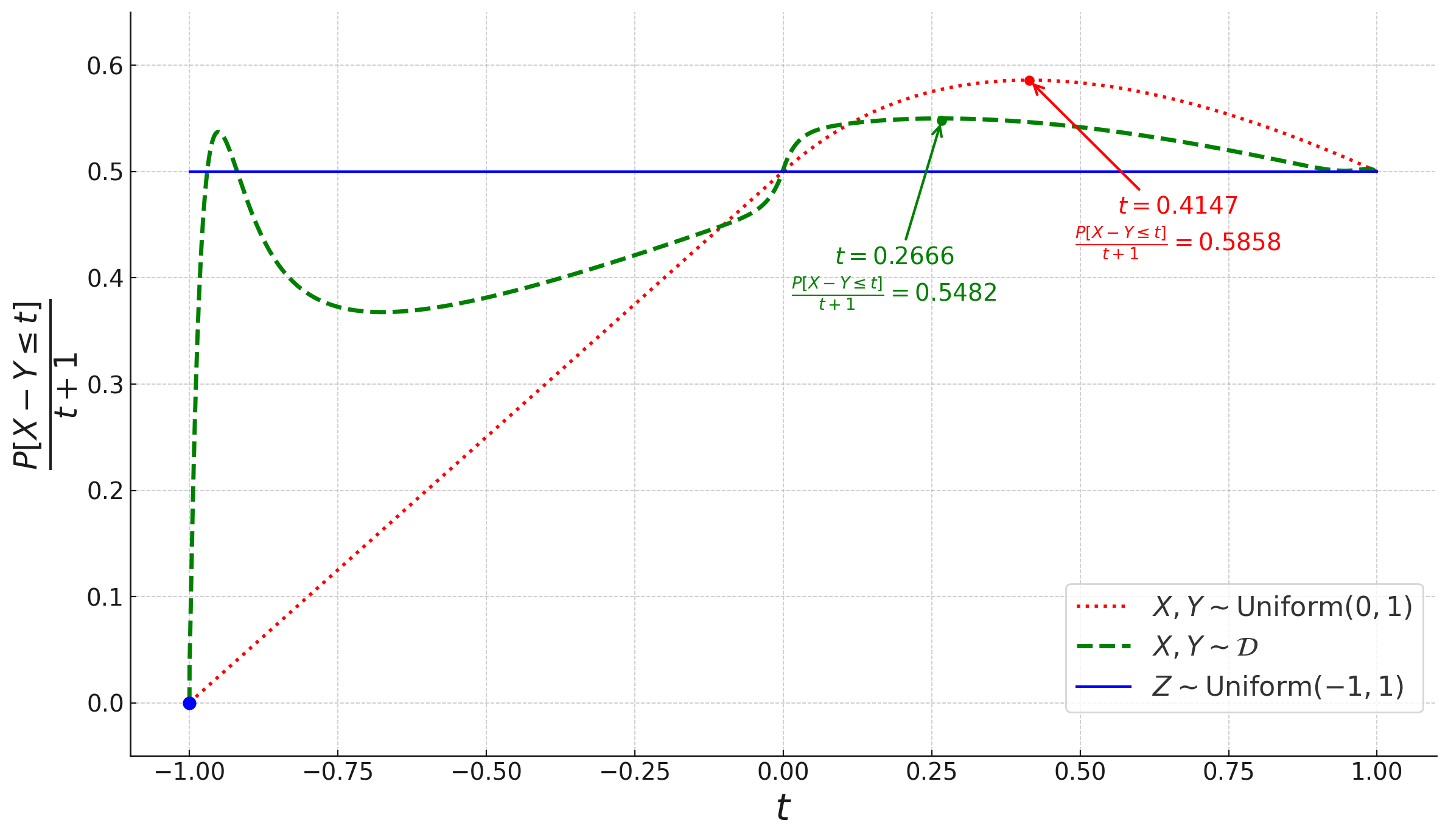}
    \caption{$\frac{\P{X-Y\leq t}}{t+1}$}
    \label{fig:2b}
  \end{subfigure}
  \caption{\cref{fig:2a} and \cref{fig:2b} compare the behavior of the difference \( X - Y \), where \( X, Y \) are sampled independently from two candidate label distributions: $\text{Uniform}(0,1)$ (red), and the refined distribution \( \mathcal{D} \) with density \( p(x) = \frac{2}{3} + \frac{23}{3}(1 - 2x)^{22} \) (green). As a reference, we also include the random variable \( Z \sim \text{Uniform}(-1,1) \) (blue), which represents a target distribution that would yield a $0.5$-approximation if realizable. While \cref{sec:lowerbound} proves that no independent label distribution over \([0,1]\) can make \( X - Y \sim \text{Uniform}(-1,1) \), the behavior of \( Z \) remains a useful benchmark. In both subfigures, the distribution \( \mathcal{D} \) yields a difference distribution that more closely mimics \( Z \) than $\text{Uniform}(0,1)$ does. In particular, as established in \cref{lem:unif} and \cref{Claim:MaxCrazyDist}, the maximum of \( \frac{P[X - Y \leq t]}{t + 1} \) occurs at \( t = 2 - \sqrt{2} \approx 0.4147 \) for $\text{Uniform}(0,1)$ and at \( t = 0.2666 \) for \( \mathcal{D} \), with corresponding values \( 0.5858 \) and \( 0.5482 \), respectively.
    }
   \label{2}
\end{figure}

   % \begin{figure}[H]
   %      \centering
   %      \includegraphics[scale=0.5]{ComparisonUniformCrazyX-Y(Uniform)Overt+1.png}
   %      \caption{}
   %      \label{2}
   %      \end{figure}
        
\section{Lower Bound}\label{sec:lowerbound}

An interesting problem is to find the function that minimizes the approximation ratio when using an \textit{independent} distribution over labels: 

\begin{problem}\label{problem} Let $\mathcal{P}$ denote the set of probability distributions supported on $[0,1]$. Find $$\alpha := \inf_{D \in \mathcal{P}} \sup_{t \in [-1,1]} \frac{\mathbb{P}_{\ell(u), \ell(v) \sim D} [\ell(u) - \ell(v) \le t]}{t+1}.$$
\end{problem}
While this problem does not capture the possibility of correlated label distributions (which we exploit in the bipartite case in \cref{sec:extras}), we find it quite interesting on its own. While some related questions have been studied in additive combinatorics, for example by Cloninger and Steinerberger \cite{CS17} and Matolcsi and Vinuesa \cite{MV10}, we were unable to find a resolution of this question in the literature. We have been able to bound:
$$0.542 < \alpha < 0.549,$$
where the upper bound is from \cref{sec:improved_alg}. Here, we show a complete proof of a lower bound of $0.539$. 
 
 \begin{claim}\label{Claim:lowerbound for alpha}
$\alpha > 0.539.$
 \end{claim}
 
 \begin{proof}
 Let $D$ be a distribution supported on $[0,1]$ and $X, Y$ be independent samples from $D$. Suppose that 
 \[
\Pr[X-Y \leq x] \leq \alpha (x+1) \quad \text{for all} \quad x \in [-1,1].
\]

Denote $Z := X-Y$. First note that symmetry of \(Z\) gives the complementary lower bound
\[
\Pr[Z \leq x] = \Pr[Z \ge -x] = 1 - \Pr[Z \le -x] \geq 1-\alpha(1-x), \quad \text{for all} \quad x \in [-1,1].
\]
These inequalities imply
\[
\\-(\alpha - \tfrac12)(1-x) \leq \Pr[Z \leq x] - \frac{x+1}{2} \leq (\alpha - \tfrac12)(1+x) \quad \text{for all} \quad x \in [-1,1]. \quad (\star)
\]

Fix some $t \in (\pi, 2\pi)$ that we choose later (so $\sin t < 0$) and observe that
\begin{align*}\mathbb{E}[\cos(tZ)] &= \mathbb{E}[\cos(t(X-Y))] \\ &= \mathbb{E}[\cos(tX)]\mathbb{E}[\cos(tY)] +  \mathbb{E}[\sin(tX)]\mathbb{E}[\sin(tY)] \\ &=\mathbb{E}[\cos(tX)]^2 +  \mathbb{E}[\sin(tX)]^2 \ge 0.\end{align*}
Integration by parts gives
\[
\mathbb{E}[\cos(tZ)] = \cos t + t \int_{-1}^{1} \sin(tx) \Pr[Z \leq x] \, dx.
\]
For \(U \sim \text{Uniform}(-1,1)\) (whose CDF is \(\frac{x+1}{2}\)), we have \(\mathbb{E}_{U \sim \text{Uniform}(-1,1)}  [\cos(tU)] = \frac{\sin t}{t}\), therefore
\[
\frac{\sin t}{t} = \cos t + t \int_{-1}^{1} \sin(tx) \cdot \left(\frac{x+1}{2}\right)\, dx.
\]
Using  $\mathbb{E}[\cos(tZ)] \ge 0$ and subtracting the two formulas above yields 
\[
 \frac{-\sin t}{t}\le \mathbb{E}[\cos(tZ)] - \frac{\sin t}{t}
= t \int_{-1}^{1} \sin(tx) \left( \Pr[Z \leq x] - \frac{x+1}{2} \right) dx.
\]
By symmetry, we have \[ \int_{-1}^{1} \sin(tx) \left( \Pr[Z \leq x] - \frac{x+1}{2} \right) dx = 2  \int_{0}^{1} \sin(tx) \left( \Pr[Z \leq x] - \frac{x+1}{2} \right) dx. \]
Since $\sin(tx)$ has only one root $\pi/t$ in $(0,1)$, we may use $(\star)$ to upper bound the integral
\[
\frac{1}{\alpha-\tfrac12}\int_{0}^{1} \sin(tx) \left( \Pr[Z \leq x] - \frac{x+1}{2} \right) dx \le \int_{0}^{\pi/t} \sin(tx)(1+x)\, dx - \int_{\pi/t}^1 \sin(tx)(1-x) \, dx,
\]
which evaluates to $\frac{3t + \sin t}{t^2}$. It follows that
\[ \alpha \ge \frac{1}{2} + \sup_{t \in (\pi, 2\pi)} \frac{-\sin t}{6t+2\sin t} \ge \frac{1}{2} + \frac{-\sin 4.5}{27+2\sin 4.5}> 0.539.\qedhere \]
 \end{proof}
 We remark that using the function $\cos(4.4X) + 0.29 \cos(8.9 X)$ instead of $\cos(4.5 X)$ in the above proof shows a slightly stronger lower bound $\alpha > 0.542$.
 
 From the \cref{Claim:lowerbound for alpha} it follows that every distribution $D \in \mathcal{P}$ has a corresponding CDF of $D-D$ which significantly exceeds the CDF of $\text{Uniform}(-1,1)$ at some point. Therefore we cannot attain a ratio of $\frac{1}{2}$. However, this may still be possible for correlated label distributions, and indeed, depending on the structure of the graph such distributions \textit{can} exist. In the next section we show they can be constructed easily for bipartite graphs. 

 %\begin{conjecture*}
 %Given a DAG $G = (V, E)$ along with thresholds $t_e$ for $e \in E$ (not necessarily related to the linear program), we can always find a labeling $L : V \to [0,1]$ such that the inequality $L(u) - L(v) \le t_{(u,v)}$ holds for at most $$\sum_{e \in E} \left(\frac{t_e + 1}{2}\right)$$ edges $(u,v) \in E$, which would imply a $\frac{k+1}{2}$-approximation algorithm for DED($k$).
 %\end{conjecture*}
 
 %It is easy to check the conjecture holds for paths and for DAGs with at most 3 vertices, even if we strengthen the upper bound to $\sum_{e \in E} \max(0, t_e)$. In general, however, we have constructions for which this stronger upper bound does not hold, and even for DAGs with $4$ vertices and arbitrary thresholds, it is not clear how to select the labels. Ideally, it should be possible to pick the label of each vertex independently at random with probability distribution somehow related to the thresholds of the edges adjacent to it.

\section{Optimal Rounding for Special Cases}\label{sec:extras}

In this section, we show how to obtain our desired bound of $0.5(k+1)$ for two special cases. The result for bipartite graphs is tight up to constants, while the second result is tight as there is a matching integrality gap on a point with $x_e = \frac{1}{k}$ for all edges. 

\subsection{Bipartite Graphs}

Here we show that for any bipartite DAG there exists a distribution over labels so that the randomized algorithm from above achieves an approximation ratio of $0.5(k+1)$. Notably, the instances from Kenkre et al. used to obtain UGC hardness of $\lfloor 0.5k \rfloor$ are bipartite, so this is tight up to constants assuming the UGC.  

\begin{theorem}
    Let $G=(V,A)$ be bipartite with bipartition $(A',B')$ and $x$ be a feasible solution to \eqref{eq:DED_LP}. Then, there is an algorithm that produces a \DED ~solution with expected cost at most $(0.5)(k+1) \cdot c(x)$. 
\end{theorem}
\begin{proof}
Our algorithm is identical to before, except now we produce a correlated distribution over labels. To produce the labeling $\ell: V \to [0,1]$, we first sample $y \in [0,1]$ uniformly. Then, with probability $\frac{1}{2}$, assign $\ell(v) = 0$ for all $v \in A'$ and $\ell(v) = y$ for all $v \in B'$. Otherwise, assign $\ell(v) = y$ for all $v \in A'$ and $\ell(v) = 0$ for all $v \in B'$. 

For every edge $a=(u,v) \in A$, the distribution over $\ell(v)-\ell(u)$ is now uniform over $[-1,1]$, as with probability $\frac{1}{2}$ we obtain a uniformly random number between $0$ and $1$ and otherwise a uniformly random number between $-1$ and $0$. Therefore, if we now run the same algorithm as before, we have:
\begin{align*}
    \Pr[\text{$a$ is cut}] &= \Pr[\ell(v)-\ell(u) \le (k+1)x_e - 1]\\
    &= \frac{(k+1)x_e - 1 + 1}{2} && \text{As $\ell(v)-\ell(u) \sim \text{Uniform}(-1,1)$}\\
    &= 0.5(k+1)x_e
\end{align*}
as desired.
\end{proof}
It is unclear whether correlated distributions can be used to improve the approximation ratio in the general case. 

\subsection{Instances with Structured LP Solutions}

Here we show that an independent distribution over labels exists when we have a structured LP solution. It implies, for example, that if our LP solution has the property that $x_e \in \{0,c\}$ for all $e \in E$ for some fixed value $c$, we can obtain a $0.5(k+1)$-approximation, but is slightly more general (for example, we can obtain the same ratio if all edges have value at least $\frac{1.5}{k+1}$). 
\begin{theorem}
If $x$ is a feasible solution to \eqref{eq:DED_LP}, and for every edge $e$ with $x_e > 0$ we have:
$$\frac{1+\frac{1}{r+1}}{k+1} \le x_e < \frac{1+\frac{1}{r}}{k+1}$$
for some integer $r \ge 1$, then there is an algorithm that produces a solution of expected cost at most $\frac{k+1}{2} c(x)$. In the case that $r=1$, the upper bound is not necessary as we may simply cut all edges with value at least $\frac{2}{k+1}$ and round the remaining edges.
\end{theorem}
\begin{proof}
    Let $\mu_r$ be the uniform distribution over labels in the set $\{0,\frac{1}{r},\frac{2}{r},\dots,1\}$. We claim applying our framework by selecting each label independently from $\mu_r$ proves the theorem. 

    We cut an edge $e=(u,v)$ if 
    $$\ell(v) - \ell(u) \le (k+1)x_e - 1 < \frac{1+\frac{1}{r}}{k+1}(k+1)-1 = \frac{1}{r}$$
    Therefore, since $\ell(u) - \ell(v)$ is either negative, 0, or at least $\frac{1}{r}$, we cut an edge exactly when $\ell(v) \le \ell(u)$. So, the probability we cut any edge with $x_e > 0$ is:
    \begin{align*}
        \Pr[\ell(v) \le \ell(u)] &= 1-\Pr[\ell(v) > \ell(u)] \\
        &= 1-\sum_{i=0}^r \Pr\left[\ell(v) = \frac{i}{r}\right] \cdot \Pr\left[\ell(u) < \frac{i}{r}\right] \\
        &= 1-\sum_{i=0}^r \frac{1}{r+1} \cdot \frac{i}{r+1} = 1-\frac{r}{2(r+1)} = \frac{r+2}{2(r+1)}
    \end{align*}
    Now, the probability we cut an edge compared to $x_e$ (our approximation ratio) is:
    $$\frac{r+2}{2(r+1)} \cdot \frac{1}{x_e} \le \frac{r+2}{2(r+1)} \cdot \frac{k+1}{1+\frac{1}{r+1}} = \frac{k+1}{2}$$
    As desired.
\end{proof}

\section{Conclusion}

The main open question remaining is whether it is possible to obtain an algorithm with approximation ratio $0.5(k+1)$, or whether this is UGC-Hard. We are unaware of any instances with integrality gap above $0.5(k+1)$, so improving this bound is also of interest. Another open problem is whether it is NP-Hard to obtain an $O(1)$ approximation for \DED. 

\bibliographystyle{plain}
\bibliography{upm}

\appendix
\section{Automated Verification}\label{sec:appendix}

We used Mathematica to verify the upper bound of $0.5482$ claimed in \cref{subsec:improved_distribution}. The code used is as follows:
\begin{verbatim}
P = ProbabilityDistribution[2/3 +(23/3)*(1 \[Minus] 2 x)^(22), {x, 0, 1}]
T = TransformedDistribution[y - z, {y \[Distributed] P, z \[Distributed] P}]
NMaximize[
  { CDF[T, t]/(t+1), -1 <= t <= 1 },
  t,                     
  WorkingPrecision -> 50,
  PrecisionGoal    -> 30,
  AccuracyGoal     -> 30
]
\end{verbatim}
Here we include a copy of a high precision run in Mathematica showing a worst case ratio of slightly below 0.5482. 
\begin{figure}[htb!]
    \centering
    \includegraphics[width=\linewidth]{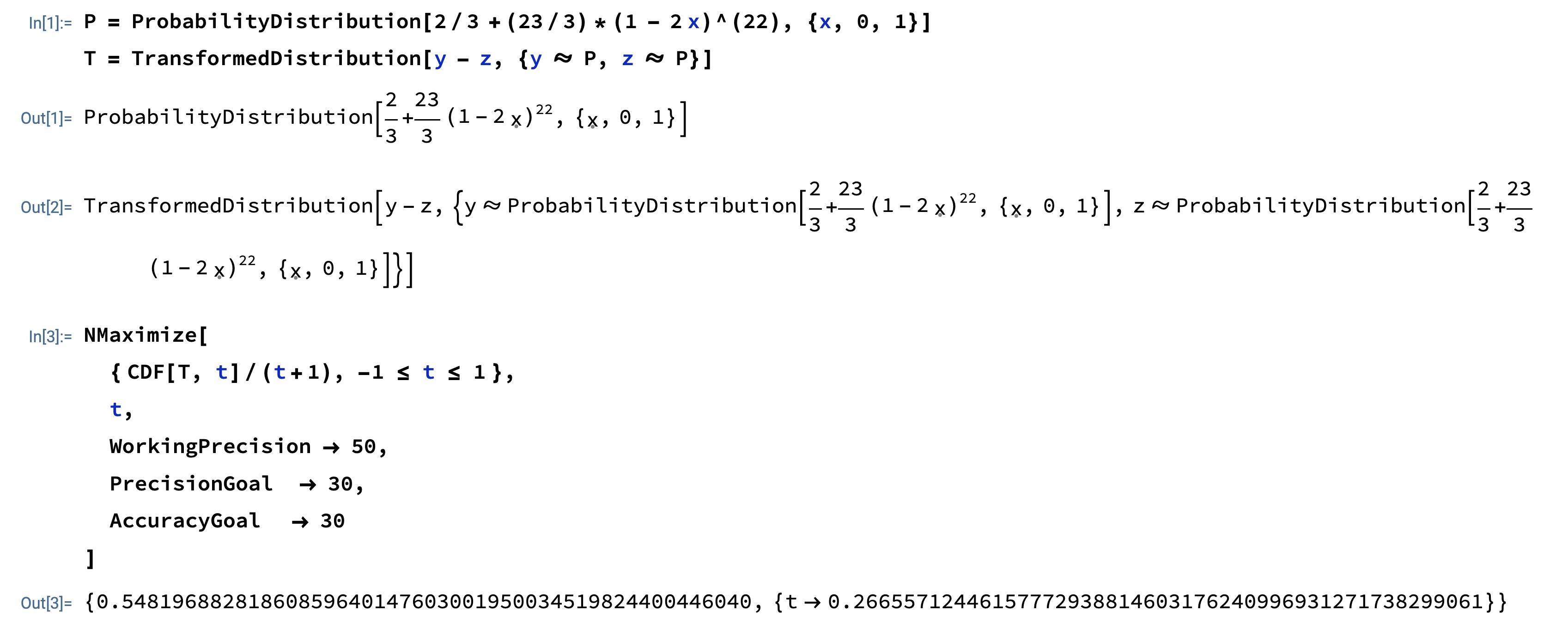}
    \caption{Execution of upper bound code.}
    \label{fig:upperbound}
\end{figure}
 \end{document}